\documentclass[12pt]{article}

\usepackage{subfigure}

\usepackage{latexsym}
\usepackage{graphics}
\usepackage{amsmath}
\usepackage{xspace}
\usepackage{amssymb}
\usepackage{psfrag}
\usepackage{epsfig}
\usepackage{amsthm}
\usepackage{pst-all}

\usepackage{color}

\definecolor{Red}{rgb}{1,0,0}
\definecolor{Blue}{rgb}{0,0,1}
\definecolor{Olive}{rgb}{0.41,0.55,0.13}
\definecolor{Yarok}{rgb}{0,0.5,0}
\definecolor{Green}{rgb}{0,1,0}
\definecolor{MGreen}{rgb}{0,0.8,0}
\definecolor{DGreen}{rgb}{0,0.55,0}
\definecolor{Yellow}{rgb}{1,1,0}
\definecolor{Cyan}{rgb}{0,1,1}
\definecolor{Magenta}{rgb}{1,0,1}
\definecolor{Orange}{rgb}{1,.5,0}
\definecolor{Violet}{rgb}{.5,0,.5}
\definecolor{Purple}{rgb}{.75,0,.25}
\definecolor{Brown}{rgb}{.75,.5,.25}
\definecolor{Grey}{rgb}{.5,.5,.5}

\newcommand{\pr}{\mathbb{P}}
\newcommand{\E}[1]{\mathbb{E}\!\left[#1\right]}
\newcommand{\R}{\mathbb{R}}

\newcommand{\mb}[1]{\mbox{\boldmath $#1$}}

\newcommand{\ignore}[1]{\relax}

\newtheorem{theorem}{Theorem}[section]

\newtheorem{lemma}[theorem]{Lemma}

\newcommand{\ER}{Erd{\"o}s-R\'{e}nyi }

\newcommand{\eps}{\epsilon}
\newcommand{\abs}[1]{\lvert #1\rvert}
\newcommand{\hJ}{\hat{J}}

\usepackage{float}

\author{
\sf David Gamarnik
\thanks{Operations Research Center and Sloan School of Management, MIT. Email: gamarnik@mit.edu} 
\and 
\sf Aukosh Jagannath
\thanks{Department of Statistics and Actuarial Science and Department of Applied Mathematics, University of Waterloo.
Email: a.jagannath@uwaterloo.ca
}
\and 
\sf Alexander S. Wein
\thanks{Department of Mathematics, Courant Institute of Mathematical Sciences,
New York University.
Email: awein@cims.nyu.edu}
}

\begin{document}

\title{Circuit Lower Bounds for the p-Spin  Optimization Problem}
\date{\today}

\maketitle

\begin{abstract}
We consider the problem of finding a near ground state of a $p$-spin model with Rademacher couplings by means of a low-depth circuit. As a direct extension of the authors' recent work~\cite{gamarnik2020lowFOCS}, we establish that any poly-size $n$-output circuit that produces a spin assignment with objective value within a certain constant factor of optimality, must have depth at least $\log n/(2\log\log n)$ as $n$ grows. This is stronger than the known state of the art bounds of the form $\Omega(\log n/(k(n)\log\log n))$ for similar combinatorial optimization problems, where $k(n)$ depends on the optimality value. For example, for the largest clique problem $k(n)$ corresponds
to the square of the size of the clique~\cite{rossman2010average}. At the same time our results are not quite comparable since in our case the circuits are required to produce a \emph{solution} itself rather than solving the associated decision problem. As in our earlier work~\cite{gamarnik2020lowFOCS}, the approach  
is based on the overlap gap property (OGP) exhibited by random $p$-spin models, but the derivation of the circuit lower bound relies further on standard facts from Fourier analysis
on the Boolean cube, in particular the Linial-Mansour-Nisan Theorem.

To the best of our knowledge, this is the first instance when methods from spin glass theory have ramifications for circuit complexity.
\end{abstract}

\section{Introduction}
Boolean circuits constitute one of the standard models for understanding algorithmic tractability 
and hardness in combinatorial optimization problems~\cite{arora2009computational,sipser,alon2004probabilistic}. 
One potential route to proving the widely-believed
conjecture $P\ne NP$ would be to show non-existence of polynomial-size Boolean circuits
solving problems in $NP$. A large body of literature in circuit complexity is devoted
to establishing limits on the power of circuits with bounds on its depth, specifically
the power of constant depth circuits known as AC[0] circuits and its immediate extension --
circuits with nearly logarithmic (in problem size) depth. Many hardness results have been
obtained by working on models with random inputs. A classical result is one by H\r{a}stad~\cite{hastad1986almost}
who established a $\log n/\left(\log\log n+O(1)\right)$ lower bound on the depth of any poly-size circuit that computes the $n$-parity function. A lot of focus is on circuit complexity for computing various
natural combinatorial optimization problems.
For example Rossman~\cite{rossman2010average} has shown
that poly-size circuits detecting the presence of a $k(n)$-size clique in a graph must have 
depth at least $\Omega(\log n/(k^2(n)\log\log n))$, where $n$ is the number of graph nodes and $k(n)$
is any function growing in $n$.
This was obtained by working with the sparse random \ER  graph model with 
the edge probability tuned so that the largest ``naturally occurring'' clique is of size smaller than $k(n)$,
with high probability as $n$ increases. In fact, for combinatorial optimization problems this 
represents the state of the art lower bound on circuit depth for polynomial-size circuits, though many extensions exist for subgraph homomorphism existence 
problems~\cite{rossman2018lower,li2017ac}.

In this paper we establish a lower bound of $\log n/(2\log\log n)$ on the depth of any polynomial-size Boolean circuit that solves another natural combinatorial optimization problem. It is a
core computational problem in statistical physics of finding a near ground state
of a p-spin model with i.i.d.\ Rademacher couplings. 
Specifically, we show that any polynomial-size $n$-output Boolean circuit
that produces a $\pm 1$ spin assignment with objective value within a certain constant factor of optimality, has depth at least $\log n/(2\log\log n)$. 
In particular, unlike the aforementioned result for cliques, our lower bound does 
not depend on the objective value. The result though, strictly speaking, does not improve upon
\cite{rossman2010average}, since our circuits are supposed to produce the \emph{entire solution}, as opposed to 
just solving
the associated YES/NO decision problem. Nevertheless, to the best of our knowledge, our result presents 
the strongest known circuit depth lower bounds for combinatorial optimization problems.
\ignore{\textcolor{red}{Something perhaps we should check with Rossman}}

Our result is obtained as a fairly direct extension
of our recent work~\cite{gamarnik2020lowFOCS}, where limits are obtained for algorithms
based on low-degree polynomials for the same problem (though with Gaussian as opposed to Rademacher couplings). 
The associated optimization problem admits a polynomial-time approximation scheme for the case 
$p=2$ \cite{montanari2021optimization} and also for some so-called mixed spin models~\cite{alaoui2020optimization}, as well as for some so-called spherical
spin glass models~\cite{subag2021following} associated with optimizing over the sphere as opposed to the binary cube.
But as shown in~\cite{gamarnik2020lowFOCS}, algorithms based on low-degree polynomials fail in general
for the problem of finding near ground states in $p$-spin models.

As in our earlier paper we employ the overlap gap property (OGP) exhibited
by this and many other combinatorial optimization and constraint satisfaction problems with random inputs.
In the present context, the OGP says that every two spin assignments that achieve value within a certain
multiplicative factor away from optimality have to either agree on many coordinates (high overlap)
or disagree on many coordinates (low overlap). More specifically, 
as in the aforementioned reference, we employ the so-called ensemble overlap gap property (e-OGP), which
we formally introduce in the next section. At the same time, employing standard circuit insensitivity
results, specifically the Linial-Mansour-Nisan Theorem~\cite{LMN}, we show that
vectors produced by the outputs of $n$ low-depth circuits can produce two spin assignments with overlap
value ruled out by the e-OGP, thus obtaining a contradiction. 
The Linial-Mansour-Nisan Theorem states that low-depth polynomial-size circuits are well 
approximated by the low-degree terms in their Fourier expansion.
The version we use in this paper is found on page 93 in~\cite{o2014analysis}. The tightest known bound of this kind
was obtained by Tal in~\cite{tal2017tight}.

We believe that the approach used in this paper can be extended to many other models exhibiting OGP
and its variants (see~\cite{gamarnik2020lowFOCS} for a more thorough treatment of the subjects).
However, additional analysis seems necessary for validating such extensions,
since many such models are defined on sparse random graphs, 
and the variant of the Linial-Mansour-Nisan Theorem for sparse graph models (see for example Lemma~9 in
\cite{furst1991improved}) 
appears unfortunately too weak to rule out depths of order $\log n/\log\log n$. 
Similarly, it is an open
question whether OGP-based methods can be used to rule out decision-type (single output) low-depth circuits
for these problems.

\section{Models, Assumptions and Results}

\subsection{$p$-spin model and its ground state}
We begin by describing the optimization problem of finding the ground state value of a $p$-spin model.
Let {$J\in \{\pm 1\}^{n^p}$ denote a $p$-tensor on $\R^n$} with $\{\pm 1\}$-valued entries. That is, 
$J=\left(J_{i_1,\ldots,i_p}, 1\le i_1,i_2,\ldots,i_p\le n\right)$ and $J_{i_1,\ldots,i_p}=\pm 1$ for each
$p$-tuple $i_1,\ldots, i_p$. For every $\sigma \in \{\pm 1\}^n$, we let
\begin{align*}
\langle J, \sigma\rangle \triangleq 
\sum_{1\le i_1,i_2,\ldots,i_p\le n} J_{i_1,i_2,\ldots,i_p}\sigma_{i_1}\sigma_{i_2}\cdots \sigma_{i_p}.
\end{align*}
The optimal value of the optimization problem 
\begin{align}\label{eq:ground-state-value}
n^{-{p+1\over  2}}\max_{\sigma \in \{\pm 1\}^n} \langle J,\sigma\rangle
\end{align}
is called the ground state energy of the $p$-spin model with coupling $J$. Any optimizer $\sigma$ achieving
the optimal value is called the ground state. 
We will assume that the entries of $J$ are i.i.d.\ equiprobable $\pm 1$ (that is, Rademacher) entries. In this setting, we denote the optimal value of~\eqref{eq:ground-state-value} by $\eta^*_{n,p}$, which is a random variable. 
The normalization by $n^{p+1\over 2}$ in the objective function is chosen so that the ground state energy is typically order 1. In particular, we recall here the following well-known concentration result. 

\begin{theorem}\label{theorem:scaling-limit}
For every $p$, there exists $\eta^*_p>0$ such that for every $\epsilon>0$ 
\begin{align*}
\pr(|\eta^*_{n,p}-\eta^*_p|\ge \epsilon)\le \exp(-\gamma n),
\end{align*}
for some $\gamma>0$ and all large enough $n$.
\end{theorem}

The existence of the limit above follows from a rather simple interpolation argument based essentially
on Gaussian comparison inequality techniques. This was shown  first by Guerra and Toninelli~\cite{GuerraTon}.
The value of $\eta^*_p$ can be computed to any desired precision using the powerful Parisi variational
representation as was done first by Parisi~\cite{parisi1980sequence}, 
and rigorously verified later by Talagrand~\cite{talagrand2006parisi}.
See Panchenko~\cite{panchenko2013sherrington} for a book treatment of the subject.

We now state the ensemble overlap gap property (e-OGP) exhibited by the $p$-spin model above. 
For this purpose we need to consider an ensemble of random tensors constructed as follows. Let $J$ and $\tilde J$ be two independent copies of the random $p$-spin model above with Rademacher entries constructed as follows. 
Let $\leq_R$ denote the (random) total order on $[n^p]$, given by drawing points from that set
uniformly at random without replacement. We call this the uniform at random order. 
Let $I_t \in [n^p]$ be the point that was drawn at time $t$.
Starting with $J_0=J$, let $J_t$ denote the random tensor obtained from $J_{t-1}$,
by resampling the entry indexed by $I_{t}$ according to the Rademacher distribution,  independently. 
Define $\tilde J = J_{n^p}$. Note that for each fixed $t$, the values $J_t,J$ and $\tilde J$ are identically distributed (but not independent). 
We now formally state the e-OGP. 

\begin{theorem}\label{theorem:e-OGP}
For every even $p\ge 4$, there exists $0<\eta_{\rm OGP,p}<\eta_p^*$ and $0<\nu_{1,p}<\nu_{2,p}<1$, 
and $\gamma>0$ 
such that for all sufficiently large $n$, 
with probability at least $1-\exp(-\gamma n)$, 
the following holds:

\begin{enumerate}

\item[(a)]
For every $0\le t_1\le t_2\le n^p$ and every pair $\sigma_1,\sigma_2\in \{\pm 1\}^n$ satisfying 
$\langle J_{t_j},\sigma_j\rangle \ge \eta_{\rm OGP,p}n^{p+1\over 2}$ for $j=1,2$, 
it is the case 
\begin{align*}
{|\langle \sigma_1,\sigma_2\rangle|\over n}\notin (\nu_{1,p},\nu_{2,p}),
\end{align*}

\item[(b)]
In the special case $t_1=0, t_2=n^p$, in fact 
\begin{align*}
{|\langle \sigma_1,\sigma_2\rangle|\over n}\le \nu_{1,p}.
\end{align*}
\end{enumerate}

\end{theorem}

The proof is given in Section~\ref{sec:pf-ogp}. Informally, e-OGP states that for every  pair of instances in the interpolated family,
including the case of identical instances (that is $t_1=t_2$),  every two solutions 
achieving multiplicative factor $\eta_{\rm OGP,p}/\eta^*$ to optimality in those
two instances are either close
or far from each other. Furthermore, for the case of independent instances (that is $t_1=0, t_2=n^p$),
only the latter is possible. Theorem~\ref{theorem:e-OGP} was proved in~\cite{gamarnik2021overlap} for the case
of Gaussian as opposed to Rademacher distribution, but its adaptation to the latter case
follows from standard universality-type arguments based on Lindeberg's approach, which can be found in the original paper~\cite{GuerraTon} which established the existence of the limit $\eta_p^*$.
For a recent work that covers many  examples of such arguments see~\cite{sen2018optimization} and for a textbook presentation of this and related spin glass results, see \cite{panchenko2013sherrington}.

The core of the proof in~\cite{gamarnik2021overlap} is the case of a single-instance (non-ensemble) OGP 
for the Gaussian case which was
done in~\cite{chen2019suboptimality}. Its adaptation to the ensemble case follows in a rather straightforward way
using the chaos property exhibited by the $p$-spin models. Adaptation to non-Gaussian distribution,
as we said, uses standard universality arguments.

\subsection{Boolean circuits}

We now turn to the discussion of Boolean circuits. Formal definitions of those can be found in most
standard textbooks on algorithms and computation~\cite{arora2009computational,sipser,alon2004probabilistic}.
Informally, these are functions from $\{0,1\}^M\to \{0,1\}^n$ obtained by considering a directed graph with $M$ input nodes which have in-degree zero, and $n$ output nodes which have out-degree zero. 
Each intermediate node corresponds to one of three standard Boolean operation 
$\vee, \wedge$ or $\neg$. The size of the circuit is the number of nodes in the associated graph and its
depth is the length of its longest directed path. Equivalently, one
may consider instead functions from $\{\pm 1\}^M\to \{\pm 1\}^n$ by adopting the appropriate logical
functions in gates. We also adopt this assumption for convenience. In particular, when $M=n^p$,
such a Boolean function maps any $p$-tensor with $\pm 1$ entries into an $n$-vector with $\pm 1$ entries.
Thus  given an $n$-output, $M=n^p$-input Boolean circuit $C$ and $p$-tensor $J$ with $\pm 1$ entries,
we denote by $C(J)$ the binary vector in $\{\pm 1\}^n$ produced when $C$ takes input $J$. We denote by
$C_j$ the single-output Boolean circuit associated with the $j$-th component of $C$ so that
$C(J)=(C_j(J), 1\le j\le n)$.

As stated in the introduction, we are interested in Boolean circuits with bounds on their size and depth.
Fix any two $n$-dependent positive integer-valued sequences $s(n),d(n)$ and a constant $0<\rho<1$. 
We denote by $\mathcal{C}_p(n,s(n),d(n),\rho)$ the family of all $n^p$-input,  $n$-output Boolean circuits $C$ which 
satisfy the following properties: 

\begin{enumerate}

\item[(a)]
The size (the number of gates) of $C$ is at most $s(n)$ and its depth 
is at most $d(n)$.

\item[(b)]
For every tensor $J\in \{\pm 1\}^{n^p}$ such
that the value in (\ref{eq:ground-state-value}) is 
non-negative, the output $C(G)$ satisfies 
\begin{align*}
\langle J, C(J)\rangle \ge \rho\max_{\sigma\in \{\pm 1\}^n} \langle J, \sigma\rangle.
\end{align*}
If the value in (\ref{eq:ground-state-value}) is negative,
the output $C(J)$ 
can be arbitrary.

\end{enumerate}
In other words, this is a family of circuits which 
is required to produce a solution with value 
which is at least a multiplicative constant $\rho$ away 
from optimality when the value of the optimization 
problem is non-negative. Clearly a family of such circuits can be empty if either $s(n)$ or $d(n)$ is too restrictive.

We now state our main result.

\begin{theorem}\label{theorem:Main-p-spin}
For every even $p\ge 4$, $\alpha>0$ and $\rho>\eta_{\rm OGP,p}/\eta^*_p$, for all sufficiently large $n$, either the associated
family of circuits $\mathcal{C}_p(n,s(n),d(n),\rho)$  is  empty, or $s(n)\ge n^\alpha$, or 
$d(n)\ge \log n/(2\log\log n)$.
\end{theorem}
The result above essentially rules out poly-size circuits with the stated bound on depth
that obtain a solution with value at least $\rho$ times the optimum value, where $\rho$ is any constant larger than $\eta_{\rm OGP,p}/\eta^*_p$. The proof is given in Section~\ref{sec:pf-main}. The constant $2$ in the depth lower bound can be slightly improved, although it appears that it has to be larger than
$1$ for our analysis to go through. We note that $\alpha$ does not appear in the claimed lower bound on the circuit depth because
the bound holds even for mildly super-polynomially sized circuits. We will not attempt to find
the optimal growth rate.

\section{Proofs}

\subsection{Proof of Theorem~\ref{theorem:Main-p-spin}}
\label{sec:pf-main}

We fix any $\alpha > 0$ and any circuit $C\in \mathcal{C}_p(n,n^\alpha,d(n),\rho)$. 
Recall that the output of the Boolean circuit $C$ acting on any input $J\in \{\pm 1\}^{n^p}$
is a vector in $\{\pm 1\}^n$ and thus has $\|\cdot\|_2^2$ norm equal $n$. 
Fix a constant $\kappa>0$.  Consider any two 
tensors $J_1,J_2 \in \{\pm 1\}^{n^p}$ which differ in at most one entry. Namely,
there exists $1\le i_1,\ldots,i_p\le n$ such that $J_1$ and $J_2$ are identical 
at all multi-indices $(i_1',\ldots,i_p')$ that are distinct from the $p$-tuple $(i_1,\ldots,i_p)$.
We say 
that the pair $(J_1,J_2)$ is $\kappa$-bad if 
\begin{equation}\label{eq:bad-defn}
\|C(J_1)-C(J_2)\|_2^2\ge \kappa n.
\end{equation}
Fix two independent copies $J$ and $\tilde J$ of a random $p$-tensor  with i.i.d.\ $\pm 1$ entries 
and consider the  discrete interpolation $J=J_0,J_1,J_2,\ldots,J_{n^p}=\tilde J$ associated with  the
statement of Theorem~\ref{theorem:e-OGP}.
The following ``stability'' result is a direct analogue of Theorem~4.2 in~\cite{gamarnik2020lowFOCS} (but for circuits instead of low-degree polynomials).

\begin{theorem}\label{theorem:no-bad}
The probability that no pair $(J_t,J_{t+1}), 0\le t\le n^p-1$ is $\kappa$-bad is at least
\begin{align}\label{eq:rate}
\exp\left(-\left(\log n\right)^{1.2d(n)}\right),
\end{align}
for all large enough $n$.
\end{theorem}

The constants $\kappa,\alpha,\eta_{\rm OGP,p}$ will be subsumed by $\log n$ as we will see.
Note that this probability converges to zero if $d(n)$ is bounded away from zero. For our purposes, we will need the rate of convergence in this bound to be no faster than exponential. (In particular, this requirement will control our depth bound.)

The proof of Theorem~\ref{theorem:no-bad} (which we include for completeness) is similar to Theorem 4.2 of~\cite{gamarnik2020lowFOCS} but with two minor changes: first we observe that \cite{gamarnik2020lowFOCS} applies not just to functions of bounded degree but more generally to functions of bounded \emph{total influence}; we then use the Linial-Mansour-Nisan Theorem to bound the total influence of any low-depth circuit.

Let's first see how Theorem~\ref{theorem:no-bad} implies the result.

\begin{proof}[Proof of Theorem~\ref{theorem:Main-p-spin}]
We fix $\epsilon>0$ with value specified later.
By Theorem~\ref{theorem:scaling-limit}, with probability at least $1-(n^p+1)\exp(-\gamma n)$, the objective values associated with  instances $J_t$ are all at least
$(1-\eps)\eta_p^*n^{p+1\over 2}$.
By Theorem~\ref{theorem:e-OGP},
with probability at least $1-\exp(-\gamma n)$ the e-OGP holds with parameters $\eta_{\rm OGP,p},~ \nu_{1,p}<\nu_{2,p}$.
Call the intersection of these two events $\mathcal{A}$ and assume the same constant $\gamma$ for convenience,
so that $\pr(\mathcal{A})\ge 1-\exp(-\gamma n)$ and the multiplier $n^p$ in the union bound is absorbed
by making $\gamma$ smaller.
Fix any $\kappa \le (\nu_{2,p}-\nu_{1,p})^2$.
Denote the event described in Theorem~\ref{theorem:no-bad}---namely the event that no pair $(J_t,J_{t+1}), 0\le t\le n^p-1$
is $\kappa$-bad---by $\mathcal{B}$.
We claim that $\mathcal{B}\subset \mathcal{A}^c$. Assuming this claim, we obtain
\begin{align*}
\exp\left(-\left(\log n\right)^{1.2d(n)}\right) \le \pr(\mathcal{B}) \le \pr(\mathcal{A}^c) \le \exp(-\gamma n),
\end{align*}
and the desired lower bound on depth $d(n)$ is obtained by changing the constant from $1.2$ to $2$ in order to subsume the term depending on $\gamma$.

It remains to prove the claim $\mathcal{B}\subset \mathcal{A}^c$. By way of contradiction, suppose that $\mathcal{B}\cap \mathcal{A}$ is non-empty.
(Note that the event $\mathcal{A}$ implies in particular that the value of the optimization problem (\ref{eq:ground-state-value}) is non-negative for every $J_t$.)
By assumption, the circuit $C$ must produce
solutions $C(J_t), 1\le t\le n^p$ with value at least a factor of $\rho$ from optimality. 
On the other hand, on the event $\mathcal{A}$, we have that the optimal value is at 
least $(1-\eps)\eta_p^* n^{p+1\over 2}.$ So if we choose $\eps$ small enough 
that $\rho (1-\eps)\eta_p^* >\eta_{{\rm OGP}, p}$, 
we have that the solutions $C(J_t)$ yield an objective value above the onset of OGP, that is, $\langle J_t,C(J_t)\rangle \ge \eta_{\rm OGP,p}n^{p+1\over 2}$.

On $\mathcal{A}$ we have by part (b) of the e-OGP that 
$\langle C(J_0),C(J_{n^p}) \rangle \le \nu_{1,p} n$. Let $1\le t_0\le n^p$ be the smallest index for which
$\langle C(J_0),C(J_{t_0}) \rangle \le \nu_{1,p} n$, and thus $\langle C(J_0),C(J_{t_0-1}) \rangle > \nu_{1,p} n$. 
(Note that this includes the possibility $t_0=1$.)
On $\mathcal{A}$, part (a) of the e-OGP applied
in the case $t_1=0,t_2=t_0-1$ implies that in fact
$\langle C(J_0),C(J_{t_0-1}) \rangle \ge\nu_{2,p} n$. By Cauchy--Schwarz,
\begin{align*}
\|C(J_0)\|_2 \cdot \|C(J_{t_0-1}) - C(J_{t_0})\|_2 &\ge |\langle C(J_0), C(J_{t_0-1}) - C(J_{t_0}) \rangle| \\
&= |\langle C(J_0), C(J_{t_0-1}) \rangle - \langle C(J_0), C(J_{t_0}) \rangle| \\
&\ge \nu_{2,p} n - \nu_{1,p} n
\end{align*}
and so, since $\|C(J_0)\|_2 = \sqrt{n}$,
\[ \|C(J_{t_0-1}) - C(J_{t_0})\|_2^2 \ge (\nu_{2,p} - \nu_{1,p})^2 n \ge \kappa n, \]
and thus $(J_{t_0-1},J_{t_0})$ is a $\kappa$-bad pair, contradicting the definition of the event $\mathcal{B}$. 
\end{proof}

Before turning to the proof of Theorem~\ref{theorem:no-bad}, let us briefly recall the following notions from Boolean analysis/Fourier analysis on $\{\pm 1\}^m$; see e.g.~\cite{o2014analysis} for a reference. (In our setting, we will always work with the specific case $m=n^p$.)
Consider the standard Fourier expansion of functions on $\{\pm 1\}^{m}$ associated with the
uniform measure on $\{\pm 1\}^{m}$. The basis of this expansion are monomials of the form
$x_S\triangleq \prod_{i\in S}x_i,\, S\subset [m]$. For every function $g:\{\pm 1\}^{m}\to \R$, the associated
Fourier coefficients are 
\begin{align*}
\hat g_S=\E{g(x) x_S}, \qquad S\subset [m],
\end{align*}
where the expectation is with respect to the uniform measure on $x=(x_1,\ldots,x_{m})\in \{\pm 1\}^{m}$. 
Then the Fourier expansion of $g$ is 
$g=\sum_{S\subset [m]} \hat g_S x_S$, and the Parseval (or Walsh) identity states that $\sum_S \hat g_S^2=\E{g(x)^2}$. 
For $i \in [m]$, let $L_i$ denote the Laplacian operator:
\begin{align*}
L_i g(x) =\sum_{S\ni i}\hat g_S x_S,
\end{align*}
and let 
$I(g)$ be the total influence
\begin{align}\label{eq:I-in-L}
I(g)= \sum_{i\in  [m]} \E{L_i g(x)^2}  
=\sum_{S \subset [m]} |S|\, \hat g_S^2. 
\end{align}
Also note that
\begin{align}
\E{L_i g(x)^2}={1\over 2}\left(\E{L_i g(x_{-i,-1})^2}+\E{L_i g(x_{-i,1})^2}\right) \label{eq:L-i-conditioned}
\end{align}
where for $\ell=\pm 1$, $x_{-i,\ell}\in \{\pm1\}^{m}$ is obtained from $x$ by fixing
the $i$-th coordinate of $x$ to $\ell$.

\begin{proof}[Proof of Theorem~\ref{theorem:no-bad}]
Fix $\delta_n=n^{-\beta}$ with some constant $\beta>p$. Let $\hat C_{S,j}, S\subset [n^p]$ be the Fourier coefficients of the function $C_j$, which 
we recall is the $j$-component of $C$. Fix a constant $c>0$ and let
\begin{align}
D_n &=\log(1/\delta_n)\left(c\log {s(n)\over \delta_n}\right)^{d(n)-1} \notag\\
&= \beta\log n\left(c(\alpha+\beta)\log n\right)^{d(n)-1}  \notag\\
&\le \left(\log n\right)^{1.1d(n)}
\label{eq:bound-Dn}
\end{align}
for all sufficiently large $n$. We use the Linial-Mansour-Nisan Theorem~\cite{LMN} (the version we take here is the version from page 93 in~\cite{o2014analysis})
which states for some universal constant $c$ the spectrum of a depth-$d(n)$ circuit with size $s(n)$
is $\delta_n$-concentrated on degree $\le D_n$. More precisely, in our context, it states that for each $j=1,2,\ldots,n$,
\begin{align*}
\sum_{S\subset [n^p] \,:\, |S|> D_n}\left(\hat C_{S,j}\right)^2\le \delta_n.
\end{align*} 
Thus for the influence function we have for each $j$
\begin{align}
I(C_j)&=\sum_{S\subset [n^p]}|S|\left(\hat C_{S,j}\right)^2 \notag \\
&\le \sum_{S \,:\, |S|\le D_n}|S|\left(\hat C_{S,j}\right)^2+n^p\delta_n \notag \\
&\le D_n\sum_{S \,:\, |S|\le D_n}\left(\hat C_{S,j}\right)^2+n^p\delta_n \notag \\
&\le D_n\sum_{S}\left(\hat C_{S,j}\right)^2+n^p\delta_n \notag \\
&= D_n+n^p\delta_n, \label{eq:Dn-bound}
\end{align}
where in the last equation we use the fact $C_j=\pm 1$ and thus $1=\E{C_j^2}=\sum_S \hat C_{S,j}^2$. 

Let $\mathcal{E}_t$ be event that the pair $(J_t,J_{t+1})$ is $\kappa$-bad (as defined in~\eqref{eq:bad-defn}). The following is similar to Lemma 4.3 in~\cite{gamarnik2020lowFOCS}.
\begin{lemma}\label{lemma:bound-kappa-bad}
The following holds for any fixed order $\le_R$ on the $n^p$ entries of the tensor:
\begin{align}\label{eq:sum-bad}
{\kappa\over 4}\sum_{0\le t\le n^p-1}\pr(\mathcal{E}_t) 
\le D_n+n^p\delta_n.
\end{align}
\end{lemma}
As $\beta>p$ we have $n^p\delta_n=o(1)$ and
we obtain that for all large enough $n$,
\begin{align}\label{eq:sum-bad-2}
{\kappa\over 4}\sum_{0\le t\le n^p-1}\pr(\mathcal{E}_t) \le 2D_n.
\end{align}

For any fixed $x\in \{\pm 1\}^{n^p}$ let $q(x)$ be the probability of the event $\cap_t \mathcal{E}_t^c$ when we condition on $J_0=x$. Namely $q(x)$ is the probability of there being no $\kappa$-bad pairs when the initial tensor of the interpolation sequence is $x$. The following is a special case of Lemma 4.4 in~\cite{gamarnik2020lowFOCS}; we will include the proof for completeness.

\begin{lemma}\label{lemma:qx-bound}
The following bound holds for any fixed order  $\leq_R$ on the $n^p$ entries of the tensor:
\begin{align*}
-\E{\log q(J)}\le 2\log 2\sum_{0\le t\le n^p-1}\pr(\mathcal{E}_t),
\end{align*}
where $J\in \{\pm 1\}^{n^p}$ is i.i.d. 
\end{lemma}

We now combine these two results to complete the proof of Theorem~\ref{theorem:no-bad}. 
Note that the probability of no $\kappa$-bad pairs is 
\begin{align*}
\pr\left(\cap_t \mathcal{E}^c_t\right)=\E{q(J)},
\end{align*}
where the expectation is with respect to the starting sample $J=J_0$ and $\leq_R$. By Lemma~\ref{lemma:qx-bound} we have
\begin{align*}
-\log \E{q(J)}\le -\E{\log q(J)}\le 2\log 2\sum_{0\le t\le n^p-1}\pr(\mathcal{E}_t)
\end{align*}
Applying (\ref{eq:sum-bad-2}) this is at most $16(\log 2) D_n/\kappa$. Exponentiating, we obtain
\begin{align*}
\pr\left(\cap_t \mathcal{E}^c_t\right)\ge \exp\left(-16(\log 2) D_n/\kappa\right).
\end{align*}
Recalling bound (\ref{eq:bound-Dn}) and since $\kappa$ is a constant, we obtain the claim by increasing
the exponent $1.1d(n)$ in (\ref{eq:bound-Dn}) to $1.2d(n)$.
\end{proof}

\begin{proof}[Proof of Lemma~\ref{lemma:bound-kappa-bad}]
Combining (\ref{eq:Dn-bound}),(\ref{eq:I-in-L})  and (\ref{eq:L-i-conditioned}) we obtain for each $j$,

\begin{align*}
D_n+n^p\delta_n\ge 
\sum_{i\in [n^p]}{1\over 2}\left(\E{L_i C_j(x_{-i,-1})^2+L_i C_j(x_{-i,1})^2}\right).
\end{align*}
Using the inequality $a^2+b^2\ge \frac{1}{2}(a-b)^2$, the right-hand side is at least
\begin{align*}
{1\over 4}\sum_{i\in [n^p]}\E{\left(L_i C_j(x_{-i,-1})-L_i C_j(x_{-i,1})\right)^2}.
\end{align*}
Note that for any $g$,
\begin{align*}
L_i g(x_{-i,-1})-L_i g(x_{-i,1})=g(x_{-i,-1})-g(x_{-i,1}).
\end{align*}
Summing over $j$ we obtain
\begin{align*}
(D_n+n^p\delta_n)n
&\ge 
{1\over 4}\sum_{i\in [n^p]}\sum_{1\le j\le n}
\E{\left(C_j(x_{-i,-1})-C_j(x_{-i,1})\right)^2} \\
&= {1\over 4}\sum_{i\in [n^p]}
\E{ \|C(x_{-i,-1})-C(x_{-i,1})\|_2^2} \\
&\ge
{1\over 4}\sum_{t\in [n^p]}\kappa n\,\pr(\mathcal{E}_t),
\end{align*}
using~\eqref{eq:bad-defn}.
\end{proof}

\begin{proof}[Proof of Lemma~\ref{lemma:qx-bound}]
Fix any order $\le_R$ with respect to which the sequence $J_0,J_1,\ldots,J_{n^p}$ is generated.
We refer to coordinates $1,2,\ldots,n^p$ as coordinates associated with this order. In particular, 
coordinate $1$ is resampled when moving from $J_0$ to $J_1$. Coordinate $2$ is resampled when
moving from $J_1$ to $J_2$, etc.

For every $0\le m\le n^p$  and tensor $x\in  \{\pm \}^{n^p}$ let $q_m(x)$ be the probability that the truncated
interpolation path $J_0,J_1,\ldots,J_m$ does not contain any $\kappa$-bad pairs $(J_t,J_{t+1}), t\le m-1$, 
when  conditioned on $J_0=x$. 
We claim the following stronger bound, from which the required claim follows as a special case $m=n^p$:
for every $m=0,1,\ldots,n^p$,
\begin{align*}
-\E{\log q_m(J_0)}\le 2\log 2\sum_{0\le t\le m-1}\pr(\mathcal{E}_t).
\end{align*}
The proof is by induction on $m$. In the base case $m=0$, the inequality above is in fact an equality with both sides equal to 0.

Assume the assertion holds for $m'\le m-1$. 
We now establish it for $m$. For any tensor $x\in \{\pm\}^{n^p}$ let $q_{1,m}(x)$ 
be the probability of the event that the 
interpolation path $J_1,J_2,\ldots,J_m$ does not contain any $\kappa$-bad pairs $(J_t,J_{t+1}), 1\le t\le m-1$, 
when  conditioned on $J_1=x$. Observe that by our inductive assumption we have
\begin{align}\label{eq:1-to-m}
-\E{\log q_{1,m}(J_1)}\le 2\log 2\sum_{1\le t\le m-1}\pr(\mathcal{E}_t).
\end{align}
For any tensor $x\in \{\pm\}^{n^p}$ let $x_{\pm}$ be the tensor obtained from $x$ by forcing coordinate
$1$ to be $\pm$. We have
\begin{align}\label{eq:qm-expanded}
-\E{\log q_m(J_0)}=-\E{(1/2)\log q_m(J_{0,+})+(1/2)\log q_m(J_{0,-})}.
\end{align}
Let $\mathcal{F}$  be the event that changing the value of the coordinate $1$ is $\kappa$-bad. Note that this
event is measurable (determined) by the realization of $J_0$, and furthermore 
$\mathcal{E}_0=\mathcal{F}$ on the event that  coordinate $1$ of $J_1$ (and therefore $J_2,\ldots,J_m$)
is different from one of $J_0$, and $\mathcal{E}_0=\emptyset$ otherwise. In particular,
$\pr(\mathcal{E}_0)=(1/2)\pr(\mathcal{F})$.

We claim
\begin{align}\label{eq:J0+}
q_m(J_{0,+})=(1/2)q_{1,m}(J_{1,+})+(1/2)q_{1,m}(J_{1,-})\mb{1}(\mathcal{F}^c).
\end{align}
We justify this identity as follows. With probability $1/2$ the  coordinate $1$ of $J_1$ is $+1$, namely
$J_1=J_{1,+}$. Conditioned on  this event
we have 
$q_m(J_{0,+})=q_{1,m}(J_{1,+})$. On the other hand, with probability $1/2$  coordinate $1$ of $J_1$ is $-1$.
In this case no bad pair occurs if  $\mathcal{F}^c$ takes place and furthermore no bad pairs occur
during the remaining $m-1$ resamplings, the probability of which is $q_{1,m}(J_{1,-})$. 

Similarly, 
\begin{align}\label{eq:J0-}
q_m(J_{0,-})=(1/2)q_{1,m}(J_{1,-})+(1/2)q_{1,m}(J_{1,+})\mb{1}(\mathcal{F}^c).
\end{align}
If the event $\mathcal{F}^c$ occurs, the right-hand sides of the
expressions (\ref{eq:J0+}) and (\ref{eq:J0-}) are identical, so 
on this event we obtain by the concavity of $\log$
\begin{align*}
(1/2)\log q_m(J_{0,+})+(1/2)\log q_m(J_{0,-})&=\log\left((1/2)q_{1,m}(J_{1,+})+(1/2)q_{1,m}(J_{1,-})\right) \\
&\ge (1/2)\log\left(q_{1,m}(J_{1,+})\right)+(1/2)\log\left(q_{1,m}(J_{1,-})\right).
\end{align*}
On the other hand, if the event $\mathcal{F}$ occurs then
\begin{align*}
\log q_m(J_{0,+})=\log\left((1/2)q_{1,m}(J_{1,+})\right),
\end{align*}
and
\begin{align*}
\log q_m(J_{0,-})=\log\left((1/2)q_{1,m}(J_{0,-})\right).
\end{align*}
In this case 
\begin{align*}
(1/2)\log q_m(J_{0,+})+(1/2)\log q_m(J_{0,-})&=
(1/2)\log\left((1/2)q_{1,m}(J_{1,+})\right)+(1/2)\log\left((1/2)q_{1,m}(J_{1,-})\right) \\
&=-\log 2+(1/2)\log\left(q_{1,m}(J_{1,+})\right)+(1/2)\log\left(q_{1,m}(J_{1,-})\right).
\end{align*}
Combining, we obtain
\begin{align*}
&(1/2)\log q_m(J_{0,+})+(1/2)\log q_m(J_{0,-}) \\
&\ge -(\log 2)\mb{1}(\mathcal{F})+(1/2)\log\left(q_{1,m}(J_{1,+})\right)+(1/2)\log\left(q_{1,m}(J_{1,-})\right).
\end{align*}

\noindent Applying (\ref{eq:qm-expanded}) we obtain
\begin{align*}
-\E{\log q_m(J_0)} &\le (\log 2)\pr(\mathcal{F})-(1/2)\E{\log q_{1,m}(J_{1,+})}-(1/2)\E{\log q_{1,m}(J_{1,-})} \\
&=(\log 2)\pr(\mathcal{F})-\E{\log q_{1,m}(J_1)}.
\end{align*}
Recalling $\pr(\mathcal{F})=2\pr(\mathcal{E}_0)$ and applying (\ref{eq:1-to-m})
we obtain the claim.
\end{proof}

\subsection{Proof of Theorem~\ref{theorem:e-OGP}}
\label{sec:pf-ogp}

The proof of Theorem~\ref{theorem:e-OGP} is similar to that in the Gaussian
case after a coupling argument where we couple the path $J_t$ above to a discretized $\hat{J}_t$ induced by 
correlated Rademacher random vectors. We explain here the key steps that are different. In the following for two mean zero random vectors we say that $X\sim_\rho Y$, if $(X_i,Y_i)$ and $(X_j,Y_j)$ are independent for $i\neq j$  and $\E{X_i Y_i}=\rho$ for each $i$. 

Let us begin with the following construction of $\leq_R$ which corresponds to said coupling. Fix $\delta>0$
so that $1/\delta$ is an integer, and let $0=\rho_0<\rho_1<\ldots<\rho_{\frac{1}{\delta}}=1$ 
be $\rho_k=k\delta$. We now construct a sequence of  
correlated instances of Rademacher random vectors as follows. Let $\hJ_0$ be an i.i.d.\ Ber(1/2) random vector.
Now for each entry of $\hJ_0$, we resample it independently with probability $\rho_1$ and leave it as it was with probability $1-\rho_1$, to obtain $\hJ_1$. Notice that $\hJ_1\sim_{1-\rho_1}\hJ_0$.
Now for those entries of $\hJ_1$ that have not been resampled, resample them with probability $(\rho_2-\rho_1)/(1-\rho_1)$. Repeat this process $1/\delta$ times to obtain a sequence $\hJ_k$ of vectors each of which has i.i.d.\ Ber(1/2) entries and are such that if $\rho_k>\rho_\ell$ then $\hJ_k \sim_{1-(\rho_k-\rho_\ell)}\hJ_\ell$. 
Furthermore, notice that as $\rho_{1/\delta}=1$, in the last step of this sequence, all of the remaining entries are resampled (independently) with probability 1, so that $\hJ_{1/\delta}$ is an independent copy of $\hJ_0$.
Now for each $k=1,\ldots,1/\delta$, 
let ${\cal I}_k$ be indices
that have been resampled when going from $\hJ_{k-1}$ to $\hJ_{k}$. We now define $\leq_R$ in the obvious way: if $I\in {\cal I}_k,I'\in{\cal I}_\ell$, then $I\leq_RI' $ if $k\leq \ell$. Within ${\cal I}_k$ we let $\leq_R$ be the uniform at random order for ${\cal I}_k$. Note that $\{{\cal I}_k\}$ is a partition of $[n^p]$. Evidently the law of $\leq_R$ is that of the uniform at random order. Notice that if we construct $(J_k)$ in the corresponding way and let $\tau_k$ be such that $\tau_0=0$ and $\tau_k-\tau_{k-1}=\abs{\mathcal{I}_k}$, then $J_{\tau_k}= \hJ_k$.

Now let us note the following continuity argument that allows us to prove the OGP statement for $(J_{\tau_k})$. 
To this end, let us note the following lemma whose proof we defer momentarily:
\begin{lemma}\label{lem:free-energy-bound}
We have that
\[
H_{t,s}(\sigma)= \frac{1}{n^{\frac{p-1}{2}}}(\langle{J_t,\sigma}\rangle-\langle{J_s,\sigma}\rangle)
\]
has 
\[
P(\max_{\sigma}\abs{H_{t,s}(\sigma)}\geq K\cdot\sqrt{t-s} n)\leq Ce^{-cn},
\]
for some $K,C,c>0$ depending only on $p$.
\end{lemma}
With this in hand, note that there are at most $O(n^{2p})$ pairs of such $t,s$, so that by a union bound and this lemma, we have that
\[
\max_{k\leq\frac{1}{\delta}}\max_{\tau_{k-1}\leq t\leq s\leq \tau_k} \abs{H_{t,s}(\sigma)}\leq K\sqrt{\delta}n.
\]   

From here, the result follows from the following theorem which is essentially Theorem~\ref{theorem:e-OGP} but for the $\rho$-correlated instances. The proof of this theorem  is found in~\cite{gamarnik2021overlap} for the case of Gaussian distribution of the entries of $J$. Its extension for the case of Rademacher distribution is obtained by standard universality type arguments and is omitted. See for example \cite{carmona2006universality,auffinger2016universality,sen2018optimization,chen2019suboptimality}  for similar arguments.

\begin{theorem}
Let $(X_I)$ be i.i.d.\ bounded random variables with 
\[
\E{X_I}=0\qquad \E{X_I^2} = 1\qquad \|X_I\|_\infty < M
\]
for some $M>0$.
Let $(X^\rho_I)$ be such that $X^\rho \sim_\rho X$. 
Let $H(\sigma)$ denote the Hamiltonian corresponding to $X$ and $H_\rho$ denote 
that corresponding to $X^\rho$.
Then for every $\eps>0$ and $\rho\in(0,1)$, there
exists $C,\tilde{\mu}>0$ such that with probability $1-\exp(-Cn)/C$, 
for every $\sigma_1,\sigma_2 \in \{\pm 1\}^n$ satisfying
$ H(\sigma_1)/n\geq \E{\eta_n} - \tilde{\mu},H^\rho(\sigma_2)/n\geq\E{\eta_n}-\tilde{\mu}$
we have that $\abs{\langle\sigma_1,\sigma_2\rangle}\leq n\eps$.
If $\rho=1$, the same result holds except we have that there are $0\leq a<b\leq 1$
such that $\abs{\langle\sigma_1,\sigma_2\rangle}\leq [0,a]\cup[b,1]$.
\end{theorem}

\begin{proof}[Proof of Lemma~\ref{lem:free-energy-bound}]
First, by symmetrization, it suffices to prove the same bound for $H_{t,s}(\sigma)$. In this case, 
note that conditionally on  $\leq_R$, the map 
\[
\hJ_t-\hJ_s\mapsto \frac{1}{n} \max_\sigma H_{t,s}(\sigma)
\]
is uniformly $\sqrt{2^p\abs{\tau_k-\tau_s}/n^{p+1}}$-Lipschitz. 
On the other hand, conditionally on $\leq_R$, we have that $Y = J_t-J_s$
is a Rademacher vector of length at most $\mathcal{I}_{t,s}$ where $\mathcal{I}_{t,s}$ are those indices resample between time $t$ and $s$.
Note that $\abs{\mathcal{I}_{t,s}}=\abs{\tau_k-\tau_s}$.
Thus applying McDairmids inequality, we have that, conditionally on $\leq_R$, this maximum, 
\[
M = \frac{1}{n}\max_\sigma H_{t,s}(\sigma)
\]
 has
\[
P(\abs{ M -\E{M}} \geq \sqrt{\abs{\tau_k-\tau_s}} \eta)\leq \exp(-2\eta^2 \abs{\tau_k-\tau_s}).
\]
Now, since $H_{t,s}$ is conditionally sub-gaussian, we may apply Talagrand's comparison inequality \cite[Cor 8.6.2]{vershynin2018high}, 
to obtain
\[
\E{M}\leq C \E{\tilde{M}}
\]
for some universal constant where $\tilde{M}$ is obtained from $M$ by replacing $(X_I,X^\rho_I)$ with gaussian random vectors correlated in the same fashion. By a standard application of Slepian's inequality, we obtain,
\[
\E{\tilde{M}} \leq C(p) \sqrt{\abs{\tau_k-\tau_s}},
\]
where we have used here that the variance of $H_{t,s}$ is at most $2(t-s)$.
Finally note that with probability $1-\exp(-cn^p)$, we have $\abs{\tau_k-\tau_s}\leq 2(t-s)$.
Combining these bounds yields the desired result.
\end{proof}

\section*{Acknowledgements} The authors are very grateful to Benjamin Rossman for educating us
on the state of the art bounds in circuit complexity. The first author gratefully acknowledges the 
support from the NSF grant DMS-2015517. 
The second author gratefully acknowledges support from NSERC; cette recherche a été financée en partie par CRSNG [RGPIN-2020-04597, DGECR-2020-00199]. 
The third author was partially supported by NSF grant DMS-1712730 and by the Simons Collaboration on Algorithms and Geometry.

\bibliographystyle{amsalpha}
\bibliography{main}

\providecommand{\bysame}{\leavevmode\hbox to3em{\hrulefill}\thinspace}
\providecommand{\MR}{\relax\ifhmode\unskip\space\fi MR }
\providecommand{\MRhref}[2]{%
  \href{http://www.ams.org/mathscinet-getitem?mr=#1}{#2}
}
\providecommand{\href}[2]{#2}
\begin{thebibliography}{CGPR19}

\bibitem[AB09]{arora2009computational}
Sanjeev Arora and Boaz Barak, \emph{Computational complexity: a modern
  approach}, Cambridge University Press, 2009.

\bibitem[AC16]{auffinger2016universality}
Antonio Auffinger and Wei-Kuo Chen, \emph{Universality of chaos and
  ultrametricity in mixed p-spin models}, Communications on Pure and Applied
  Mathematics \textbf{69} (2016), no.~11, 2107--2130.

\bibitem[AMS20]{alaoui2020optimization}
Ahmed~El Alaoui, Andrea Montanari, and Mark Sellke, \emph{Optimization of
  mean-field spin glasses}, arXiv preprint arXiv:2001.00904 (2020).

\bibitem[AS04]{alon2004probabilistic}
Noga Alon and Joel~H Spencer, \emph{The probabilistic method}, John Wiley \&
  Sons, 2004.

\bibitem[CGPR19]{chen2019suboptimality}
Wei-Kuo Chen, David Gamarnik, Dmitry Panchenko, and Mustazee Rahman,
  \emph{Suboptimality of local algorithms for a class of max-cut problems}, The
  Annals of Probability \textbf{47} (2019), no.~3, 1587--1618.

\bibitem[CH06]{carmona2006universality}
Philippe Carmona and Yueyun Hu, \emph{Universality in
  {Sherrington--Kirkpatrick's} spin glass model}, Annales de l'Institut Henri
  Poincare (B) Probability and Statistics, vol.~42, Elsevier, 2006,
  pp.~215--222.

\bibitem[FJS91]{furst1991improved}
Merrick~L Furst, Jeffrey~C Jackson, and Sean~W Smith, \emph{Improved learning
  of {AC0} functions}, COLT, vol.~91, 1991, pp.~317--325.

\bibitem[GJ21]{gamarnik2021overlap}
David Gamarnik and Aukosh Jagannath, \emph{The overlap gap property and
  approximate message passing algorithms for $ p $-spin models}, The Annals of
  Probability \textbf{49} (2021), no.~1, 180--205.

\bibitem[GJW20]{gamarnik2020lowFOCS}
David Gamarnik, Aukosh Jagannath, and Alexander~S Wein, \emph{Low-degree
  hardness of random optimization problems}, 61st Annual Symposium on
  Foundations of Computer Science, 2020.

\bibitem[GT02]{GuerraTon}
F.~Guerra and F.~L. Toninelli, \emph{The thermodynamic limit in mean field spin
  glass models}, Commun. Math. Phys. \textbf{230} (2002), 71--79.

\bibitem[H{\r{a}}s86]{hastad1986almost}
Johan H{\r{a}}stad, \emph{Almost optimal lower bounds for small depth
  circuits}, Proceedings of the eighteenth annual ACM symposium on Theory of
  computing, 1986, pp.~6--20.

\bibitem[LMN93]{LMN}
Nathan Linial, Yishay Mansour, and Noam Nisan, \emph{Constant depth circuits,
  fourier transform, and learnability}, Journal of the ACM (JACM) \textbf{40}
  (1993), no.~3, 607--620.

\bibitem[LRR17]{li2017ac}
Yuan Li, Alexander Razborov, and Benjamin Rossman, \emph{On the {AC0}
  complexity of subgraph isomorphism}, SIAM Journal on Computing \textbf{46}
  (2017), no.~3, 936--971.

\bibitem[Mon21]{montanari2021optimization}
Andrea Montanari, \emph{Optimization of the {Sherrington--Kirkpatrick}
  hamiltonian}, SIAM Journal on Computing (2021), FOCS19--1.

\bibitem[O'D14]{o2014analysis}
Ryan O'Donnell, \emph{Analysis of boolean functions}, Cambridge University
  Press, 2014.

\bibitem[Pan13]{panchenko2013sherrington}
Dmitry Panchenko, \emph{The {Sherrington-Kirkpatrick} model}, Springer Science
  \& Business Media, 2013.

\bibitem[Par80]{parisi1980sequence}
Giorgio Parisi, \emph{A sequence of approximated solutions to the {SK} model
  for spin glasses}, Journal of Physics A: Mathematical and General \textbf{13}
  (1980), no.~4, L115.

\bibitem[Ros10]{rossman2010average}
Benjamin Rossman, \emph{Average-case complexity of detecting cliques}, Ph.D.
  thesis, Massachusetts Institute of Technology, 2010.

\bibitem[Ros18]{rossman2018lower}
\bysame, \emph{Lower bounds for subgraph isomorphism}, Proceedings of the
  International Congress of Mathematicians: Rio de Janeiro 2018, World
  Scientific, 2018, pp.~3425--3446.

\bibitem[Sen18]{sen2018optimization}
Subhabrata Sen, \emph{Optimization on sparse random hypergraphs and spin
  glasses}, Random Structures \& Algorithms \textbf{53} (2018), no.~3,
  504--536.

\bibitem[Sip97]{sipser}
M.~Sipser, \emph{Introduction to the theory of computability}, PWS Publishing
  Company, 1997.

\bibitem[Sub21]{subag2021following}
Eliran Subag, \emph{Following the ground states of full-{RSB} spherical spin
  glasses}, Communications on Pure and Applied Mathematics \textbf{74} (2021),
  no.~5, 1021--1044.

\bibitem[Tal06]{talagrand2006parisi}
Michel Talagrand, \emph{The {Parisi} formula}, Annals of mathematics (2006),
  221--263.

\bibitem[Tal17]{tal2017tight}
Avishay Tal, \emph{Tight bounds on the fourier spectrum of {AC0}}, 32nd
  Computational Complexity Conference (CCC 2017), Schloss
  Dagstuhl-Leibniz-Zentrum fuer Informatik, 2017.

\bibitem[Ver18]{vershynin2018high}
Roman Vershynin, \emph{High-dimensional probability: An introduction with
  applications in data science}, vol.~47, Cambridge University Press, 2018.

\end{thebibliography}

\end{document}